\documentclass{article}
\usepackage{amsmath}
\usepackage{amsfonts}
\usepackage{amsthm}
\usepackage{comment}
\usepackage[utf8]{inputenc}
\usepackage[toc]{appendix}
\usepackage{algorithm}
\usepackage{algpseudocode}
\usepackage{booktabs}  

\newtheorem{lemma}{Lemma}
\newtheorem{theorem}{Theorem}

\newtheorem{corollary}{Corollary}
\newtheorem{definition}{Definition}

\newtheorem{proposition}{Proposition} 
\newtheorem{remark}{Remark} 

\usepackage[
backend=biber,
style=numeric,
sorting=nty
]{biblatex}

\newcommand{\rmc}{\mathrm{RM}}
\newcommand{\nl}{\mathrm{nl}}
\newcommand{\ml}{\mathrm{ml}}
\newcommand{\hdist}{\mathrm{d}}
\newcommand{\wt}{\mathrm{wt}}
\newcommand{\crrm}{\rho}
\newcommand{\MAXVAL}{m}
\newcommand{\diag}{\mathrm{diag}}
\newcommand{\agl}{\mathrm{AGL}}

\title{The Covering Radius of the Third-Order Reed-Muller Code RM(3,7) is 20}
\author{
Jinjie Gao\footnote{School of Computer Science, Fudan University, Shanghai 200433, China. Email: jjgao18@fudan.edu.cn},
Haibin Kan
\footnote{Shanghai Key Laboratory of Intelligent Information Processing, School of Computer Science, Fudan University, Shanghai 200433, China;
Shanghai Engineering Research Center of Blockchain, Shanghai 200433, China;
Yiwu Research Institute of Fudan University, Yiwu City 322000, China. Email: hbkan@fudan.edu.cn}
Yuan Li\footnote{School of Computer Science, Fudan University, Shanghai 200433, China. Email: yuan\_li@fudan.edu.cn} and 
Qichun Wang\footnote{School of Computer Science and Technology, Nanjing Normal University, China. Email: qcwang@fudan.edu.cn}}

\begin{document}
\date{}
\maketitle

\abstract{
We prove the covering radius of the third-order Reed-Muller code
RM(3,7) is 20, which was previously known to be between 20 and 23 (inclusive).

The covering radius of RM(3, 7) is the maximum third-order nonlinearity among all 7-variable Boolean functions. It was known that there exist 7-variable Boolean functions with third-order nonlinearity 20. We prove the third-order nonlinearity cannot achieve 21. According to the classification of the quotient space of RM(6,6)/RM(3,6), we classify all 7-variable Boolean functions into 66 types. Firstly, we prove 62 types (among 66) cannot have third-order nonlinearity 21; Secondly, we prove function of the remaining 4 types can be transformed into a type (6, 10) function, if its third-order nonlinearity is 21; Finally, we transform type (6, 10) functions into a specific form, and prove the functions in that form cannot achieve third-order nonlinearity 21 (with the assistance of computers). 

By the way, we prove that the affine transformation group over any finite field can be generated by two elements.
}

\section{Introduction}

Determining the covering radius of the Reed-Muller codes is a notoriously difficult task. The \emph{covering radius} of codes is defined as the smallest integer $\rho$ such that all vectors in the vector space are within Hamming distance $\rho$ from some codeword. Reed-Muller codes, denoted by $\mathrm{RM}(r,n)$, $0 \le r \le n$, is the set of all Boolean functions in $n$ variables with degree at most $r$. The covering radius of $\mathrm{RM}(r,n)$, denoted by $\crrm(r,n)$, is exactly the maximum $r$th-order nonlinearity among all $n$-variable Boolean functions.

When $r$ is close to $n$, it is well known that $\rho(n-1, n) = 1$ and $\rho(n-2, n) = 2$. McLoughlin proves \cite{Mcl79} that
\begin{equation}
  \crrm(n-3,n)=
  \begin{cases}
    n+2, & \text{when $n$ is even}.\\
    n+1, & \text{when $n$ is odd}.
  \end{cases}
\end{equation}

For even $n$, the covering radius of $\mathrm{RM}(1,n)$ is $2^{n-1} - 2^{\frac{n}{2}-1}$; Boolean functions whose nonlinearity achieves this bound are called \emph{bent} functions \cite{Carlet2021boolean}. 

For odd $n$, by concatenating two bent functions, we have  $\crrm(1, n) \ge 2^{n-1} - 2^{\frac{n-1}{2}}$, which is the \emph{bent concatenation} bound. Hou \cite{Hou97} proves that 
\[
\crrm(1, n) \le 2 \lfloor 2^{n-2} - 2^{\frac{n-4}{2}} \rfloor
\]
Schmidt \cite{schmidt2019} proves an asymptotic result $\crrm(1,n) = 2^{n-1} - 2^{\frac{n}{2}-1}(1+o(1))$.

For specific odd $n$, Berlekamp and Welch prove  $\crrm(1, 5) = 12$ \cite{BW72}. Mykkeltveit proves $\crrm(1, 7) = 56$ \cite{Myk80}, and Hou gives a simpler proof \cite{Hou96}. Kavut and Y{\"u}cel prove $\crrm(1,9) \ge 242$ by finding a 9-variable Boolean function with nonlinearity 242 among generalized rotation symmetric Boolean functions \cite{KY10}. We direct interested readers to Carlet's book \cite{Carlet2021boolean} for more results on first-order nonlinearities, and other related topics, which is an excellent treatment on this subject.


For second-order nonlinearity, Schatz proves $\crrm(2,6) = 18$ \cite{Schatz81}. Wang proves $\crrm(2, 7) = 40$ in \cite{Wang19}. When $n \ge 8$, the exact value of $\crrm(2, n)$ is unknown.

For third-order nonlinearity, Hou \cite{Hou93} proves $20 \le \crrm(3, 7) \le 23$.
Wang \emph{et al.} \cite{LWS19, WTP18} prove that the covering radius of $\mathrm{RM}(3,7)$ in $\mathrm{RM}(4, 7)$ is 20 , the covering radius of $\mathrm{RM}(3,7)$ in $\mathrm{RM}(5, 7)$ is 20.

For higher order nonlinearity, Dougherty, Mauldin and Tiefenbruck \cite{DMT21} prove that the covering radius of $\mathrm{RM}(4, 8)$ in $\mathrm{RM}(5, 8)$ is 26 and the covering radius of $\mathrm{RM}(5, 9)$ in $\mathrm{RM}(6, 9)$ is between 28 and 32. We would like to point out that almost all the aforementioned concrete results rely on computers to some extent.

Using probabilistic methods or counting argument, one can easily prove
$
\crrm(r, n) \ge 2^{n-1} - 2^{\frac{n-1}{2}} \sqrt{{n \choose \le r}}.
$
By deriving bounds on character sums, Carlet and Mesnager \cite{CM06} prove that 
$
\crrm(r,n) \le 2^{n-1}-\frac{\sqrt{15}}{2}(1+\sqrt{2})^{r-2} \cdot 2^{\frac{n}{2}} + O(n^{r-2}).
$ In a recent work \cite{MO2021}, Mesnager and Oblaukhov provide a classification of the codewords of weights 16 and 18 for $\rmc(n-3, n)$, which could help improve the upper bounds of $\rho(r, n)$. 

In the computational complexity world, exhibiting an \emph{explicit} function with large high-order nonlinearity is an outstanding open problem. For example, nobody can prove there is a function $f$ in NP such that $\log_2 n$-order nonlinearity is at least $(1-\frac{1}{n})\cdot 2^{n-1}$ \cite{Viola2009}. 

For $\mathrm{mod}_3 : \{0,1\}^n \to \{0, 1\}$ function, where $\mathrm{mod}_3(x) = 1$ if and only if the Hamming weight of $x$ is a multiple of 3, Smolensky \cite{Smolensky1987algebraic} proves that the $r$th-order nonlinearity of $\mathrm{mod}_3$ is at least $\frac{1}{3} \cdot 2^{n-1}$, where $r = \epsilon \sqrt{n}$ for some absolute constant $\epsilon > 0$. 
For the \emph{generalized inner product} function 
\[
\mathrm{GIP}_k(x_1, x_2, \ldots, x_n) = \prod_{i=1}^k x_i + \prod_{i=k+1}^{2k} x_i + \ldots + \prod_{i=n-k+1}^n x_i,
\]
Babai, Nisan and Szegedy prove that $\nl_r(\mathrm{GIP}_{r+1}) \ge 2^{n-1}(1-\exp(-\Omega(\frac{n}{r4^r}))$ for any $r$ and $n$ \cite{BNS92}. Note that when $r = o(\log n)$, $\nl_r(\mathrm{GIP}_{r+1}) = 2^{n-1}(1 - o(1))$.

While revising this paper, we notice that Gillot and Langevin propose a descending algorithm to compute the classification of $\rmc(7,7)/\rmc(3,7)$ \cite{gillot2022classification1,gillot2022classification}. It turns out that, among 3486 non-equivalent cosets in $\rmc(7,7)/\rmc(3,7)$, the maximum third-order nonlinearity is 20, which proves $\crrm(3,7)=20$ in a different way. Recently, Gillot and Langevin prove that the covering radius of $\rmc(4,8)$ in $\rmc(6,8)$ is 26 by computing the classification of the $\rmc(6,8)/\rmc(4,8)$ \cite{gillot2022covering}.

\subsection{Our results}

\begin{theorem}
$\crrm(3, 7) = 20$	.
\end{theorem}

Determining the covering radius of third-order Reed-Muller code $\rmc(3, 7)$ is an important problem. For example, this number is closely related to the problem of the decoding $\rmc(3, 7)$ beyond the minimal distance. Besides that, determining the $\crrm(3, 7)$ would also improve other upper bounds on the covering radius of Reed-Muller codes.

Our proof is an exhaustive search \emph{in essence}, that is, by proving $\mathrm{nl}_3(f) > 20$ cannot hold for all $7$-variable Boolean functions, where $\mathrm{nl}_3(f)$ denotes the third-order nonlinearity for $f$. Changing the value of one point of $f$ results in changing $\nl_3(f)$ by at most one. Without loss of generality, it suffices to prove $\mathrm{nl}_3(f) = 21$ cannot hold.

Since there are $2^{2^7} = 2^{128}$ 7-variable Boolean functions, a brute-force search does not work; we would need to do it cleverly. We list the following major ingredients in our proof.
\begin{itemize}
\item A 7-variable Boolean function is the concatenation of two 6-variable Boolean functions, denoted by $f = f_1 \| f_2$.
\item We rely on a classification of $\mathrm{RM}(6,6) / \mathrm{RM}(3, 6)$ under affine equivalence, where there are 11 cosets that are explicitly described. We classify all 7-variable Boolean functions into $66$ types accordingly.
\item For 62 types (among 66), we prove they cannot have third-order nonlinearity 21.
\item For the remaining 4 types, we prove either their third-order nonlinearity is at most 20, or they are affine equivalent to a type (6, 10) function.
\item For type (6, 10) Boolean functions, we transform them into a specific form, while preserving its third-order nonlinearity. With the assistance of computers, we prove functions in the above-mentioned form cannot achieve third-order nonlinearity 21, and thus completes the proof of $\rho(3, 7) \le 20$. All the codes are available online \footnote{\url{https://github.com/jsliyuan/cr37}}.
\end{itemize}

Wang \emph{et al.} use similar techniques to prove that  $\crrm(2, 7) = 40$ \cite{Wang19}, and to prove the covering radius of $\rmc(3, 7)$ in $\rmc(4, 7)$ or $\rmc(5, 7)$ is 20 \cite{LWS19, WTP18}. We improve Wang's techniques to fully determine the covering radius of $\rmc(3, 7)$.

Using inequality $\crrm(r, n) \le \crrm(r-1, n-1) + \crrm(r, n-1)$, we improve the upper bound of $\crrm(3, n)$ and $\crrm(4, n)$ for $n = 8, 9, 10$.

\begin{corollary} $\crrm(3, 8) \le 60$, $\crrm(3, 9) \le 156$, 
$\crrm(3, 10) \le 372$, $\crrm(4, 8) \le 28$, $\crrm(4, 9) \le 88$, 
$\crrm(4, 10) \le 244$.
\end{corollary}

What are the limitations of the techniques? Can we go beyond $\crrm(3, 7)$? Our proof heavily relies on the classification of $\rmc(6,6)/\rmc(3, 6)$. The number of affine inequivalent classes in $\rmc(r, n)/\rmc(s, n)$ is usually enormous unless $n$ is small \cite{Hou95, Hou2021number}. However, in some cases, we do have nice classifications, for example, there are only 12 non-equivalent classes in $\rmc(3,7)/\rmc(2, 7)$; 32 in $\rmc(3,8)/\rmc(2, 8)$, which are all explicitly described \cite{hou1996gl}. This raises hope for understanding the covering radius of Reed-Muller codes, at least for those specific cases.

In addition, we prove  $\rho(r,n) \le \max_{g \in \mathcal{B}_{n-1}}\{\nl_r(g)+\nl_{r-1}(g)\}$. Note that this bound is better than $\rho(r, n) \le \rho(r, n-1) + \rho(r-1, n-1)$. For most $r$ and $n$, it is likely that $\nl_r(f)$ and $\nl_{r-1}(f)$ cannot be maximum simultaneously. So we believe our bound can potentially be used to improve other upper bounds on $\rho(r, n)$.

By the way, we prove that the affine transformation group over any finite field can be generated by two explicitly-defined elements. This result is used in our verification algorithm for proving $\rho(3, 7) \le 20$.

\section{Preliminaries}

Let $\mathbb{N} := \{0, 1, \ldots, \}$. Let $\mathbb{F}_2$ be the finite field of size 2. An $n$-variable Boolean function is a mapping from $\mathbb{F}_2^n$ to $\mathbb{F}_2$. 
Denote by $\mathcal{B}_n$ the set of all $n$-variable Boolean functions. 
Any $n$-variable Boolean function can be written as a unique multilinear polynomial in $\mathbb{F}_2[x_1, x_2, \ldots, x_n]$, say,
\[
f(x_1, \ldots, x_n) = \sum_{S \subseteq [n]} c_S \prod_{i \in S} x_i,
\]
which is called the \emph{algebraic normal form} (ANF). The \emph{algebraic degree} of $f$, denoted by $\mathrm{deg}(f)$, is the number of variables in the highest order term with nonzero coefficient. Denote by $\mathcal{H}_n^{(r)}$ the set of all $n$-variable homogeneous Boolean functions of degree $r$.

A Boolean function is \emph{affine} if its algebraic degree is at most 1. A Boolean function $f$ is \emph{linear} if $f$ is affine and $f(0, \ldots, 0) = 0$.

The \emph{Hamming weight} of a vector $x \in \mathbb{F}_2^n$, denoted by $\wt(x)$, is the number of nonzero coordinates. The \emph{weight} of a Boolean function $f$, denoted by $\wt(f)$, is the cardinality of the set $\{ x \in \mathbb{F}_2^n : f(x) = 1\}$. The \emph{distance} between two functions $f$ and $g$ is the cardinality of the set $\{ x \in \mathbb{F}_2^n : f(x) \not= g(x)\}$, denoted by $\mathrm{d}(f, g)$.

The \emph{nonlinearity} of an $n$-variable Boolean function $f$, denoted by $\mathrm{nl}(f)$, is the minimum distance between $f$ and an affine function, i.e.,
\[
\mathrm{nl}(f) := \min_{\deg(g) \le 1} \mathrm{d}(f, g).
\]
For $1 \le r \le n$, the $r$th-order nonlinearity of an $n$-variable Boolean function $f$, denoted by $\mathrm{nl}_r(f)$, is the minimum distance between $f$ and functions with degree at most $r$, i.e.,
\[
\mathrm{nl}_r(f) := \min_{\deg(g) \le r} \mathrm{d}(f, g).
\]

For integers $0 \le r \le n$, denote the $r$th-order Reed-Muller codes by $\mathrm{RM}(r, n)$, which has message length $\sum_{i \le r} {n \choose i}$ and codeword length $2^n$. Reed-Muller codes $\mathrm{RM}(r, n)$ can be viewed as the vector space of all $n$-variable Boolean functions with degree at most $r$. The \emph{covering radius} of $\mathrm{RM}(r, n)$ is
\[
\crrm(r, n) := \max_{f \in \mathcal{B}_n} \min_{g \in \mathrm{RM}(r,n)} \mathrm{d}(f,g),
\]
which is exactly $\max_{f \in \mathcal{B}_n} \mathrm{nl}_r(f)$.

Let $\mathrm{GL}(n) = \mathrm{GL}(n, \mathbb{F}_2)$ denote the \emph{general linear group} over $\mathbb{F}_2$, that is, the set of all $n \times n$ invertible matrices with the operation of matrix multiplication. Transformation $L:\mathbb{F}_2^n \to \mathbb{F}_2^n$ is called an \emph{affine transformation} if $L(x) = Ax+b$ for some $A \in \mathrm{GL}(n)$ and $b \in \mathbb{F}_2^n$. Denote by $f \circ L = f(L(x))$, and $f \circ L_1 \circ L_2 = (f \circ L_1) \circ L_2$. All affine transformations form a group, denoted by $\mathrm{AGL}(n) = \mathrm{AGL}(n, \mathbb{F}_2)$, with the operation  
\[
L_1 \circ L_2 := L_1(L_2(x)) = A_1 A_2x + A_1b_2 + b_1,
\]
where $L_1(x) = A_1 x + b_1$, $L_2(x) = A_2 x + b_2$.

Let $\mathbb{F}_q$ denote the finite field of size $q$. The general linear group $\mathrm{GL}(n)$ over the finite field $\mathbb{F}_q$ can be generated by two explicitly defined elements \cite{Wat89}. Based on the results by Waterhouse \cite{Wat89}, we prove the general affine group over any finite field can be generated by two elements in Section \ref{appendix_generators}.

Two $n$-variable Boolean functions $f, g$ are called \emph{affine equivalent} if there exists $L \in \mathrm{AGL}(n)$ such that $f \circ L = g$. Denote by $\rmc(r, n)/\rmc(s, n)$ the quotient space consisting of all cosets of $\rmc(s, n)$ in $\rmc(r, n)$, where $ s<r\le n$. $\mathrm{AGL}(n)$ acts on $\rmc(r, n)/\rmc(s, n)$ in a nature way. Cosets of $\rmc(s, n)$ in the same $\mathrm{AGL}(n)$-orbit of $\rmc(r, n)/\rmc(s, n)$  are said to be \emph{affine equivalent}.

Let $f_1, f_2$ be two $n$-variable Boolean functions. Denote the \emph{concatenation} of $f_1$ and $f_2$ by $f_1 \| f_2$ , i.e.,
\[
f_1 \| f_2 := (x_{n+1}+1)f_1 + x_{n+1}f_2.
\]
In other words,
\begin{equation*}
  (f_1 \| f_2)(x_1, \ldots, x_{n+1})=
  \begin{cases}
    f_1(x_1, \ldots, x_n), & \text{if $x_{n+1} = 0$}.\\
    f_2(x_1, \ldots, x_n), & \text{if $x_{n+1} = 1$}.
  \end{cases}
\end{equation*}

The following theorem is a classification of $\rmc(6,6)/\rmc(3,6)$. Hou has given a formula computing the number of non-equivalent cosets in $\rmc(r,n) / \rmc(s, n)$ for any $r, s, n$ \cite{Hou95}, and determines the number of non-equivalent cosets in $\rmc(6,6) / \rmc(3, 6)$ is 11. Langevin has a website containing the classifications of $\rmc(r, n)/\rmc(s, n)$ for all $s < r \le n \le 6$ \cite{Phi09}.

\begin{theorem} (Classification of $\rmc(6,6)/\rmc(3,6)$ \cite{Phi09, Mai91})
\label{thm:rmc3666_class}
Let
\begin{equation*} 
\begin{aligned}	
    & fn_0=0\\
	& fn_1=x_1x_2x_3x_4\\
	& fn_2=x_1x_2x_4x_5+x_1x_2x_3x_6\\
	& fn_3=x_2x_3x_4x_5+x_1x_3x_4x_6+x_1x_2x_5x_6\\
	& fn_4=x_1x_2x_3x_4x_5\\
	& fn_5=x_1x_2x_3x_4x_5+x_1x_2x_3x_6\\
	& fn_6=x_1x_2x_3x_4x_5+x_1x_3x_4x_6+x_1x_2x_5x_6 \\
	& fn_7=x_1x_2x_3x_4x_5x_6\\
	& fn_8=x_1x_2x_3x_4x_5x_6+x_1x_2x_3x_4\\
	& fn_9=x_1x_2x_3x_4x_5x_6+x_1x_2x_4x_5+x_1x_2x_3x_6\\
	& fn_{10}=x_1x_2x_3x_4x_5x_6+x_2x_3x_4x_5+x_1x_3x_4x_6+x_1x_2x_5x_6
\end{aligned}
\end{equation*}
be in $\rmc(6, 6)$. Then $fn_i + \rmc(3, 6)$, $i = 0, 1, \ldots, 10$, are representatives of the $\mathrm{AGL}(6)$-orbits in $\rmc(6,6) / \rmc(3, 6)$.
\end{theorem}
\begin{remark} Note that
\begin{eqnarray*}
fn_9 & = & x_1x_2x_3x_4x_5x_6 + fn_2, \\
fn_{10} & = & x_1x_2x_3x_4x_5x_6 + fn_3.
\end{eqnarray*}
We will use this fact later.
\end{remark}

From the Table \ref{table:fn10_nl3} in Appendix \ref{appendix_properties_fni}, one can see different $fn_i$ and $fn_j$ have different algebraic degrees or third-order nonlinearities. Since affine transformation preserves degree and third-order nonlinearity, cosets $fn_i + \rmc(3, 6)$, $i = 0, 1, \ldots, 10$, are not affine equivalent.

 In~\cite{LWS19}, the authors prove a recursive relation of $\mathrm{nl}_r(f)$. Using this recursive relation, one can compute the $r$th-order nonlinearity more efficiently, compared with the straightforward approach that computes the Hamming distance with all codewords in $\rmc(r, n)$.
 
\begin{lemma} (\cite{LWS19})
	Let $f = f_1 \| f_2$ be an $n$-variable Boolean function. We have
\begin{equation*} 
\mathrm{nl}_r(f)=\min_{g\in \mathcal{H}_{n-1}^{(r)} \cup \{0\}}\{\mathrm{nl}_{r-1}(f_1+g)+\mathrm{nl}_{r-1}(f_2+g)\} ,
\end{equation*}
where $\mathcal{H}_{n-1}^{(r)}$ denotes the set of all $(n-1)$-variable homogeneous Boolean functions of degree $r$. 
\end{lemma}

In order to compute $r$th-order nonlinearity of $f = f_1 \| f_2$, it suffices to compute $(r-1)$th-order nonlinearities of $f_i + g$, where $i\in \{1,2\}$ and $g$ ranges over all homogeneous degree-$r$ functions. This motivates the following definition.

 \begin{definition} \label{def:Ff} (\cite{WTP18}) Let $2 \le r \le n$. Given $f\in \mathcal{B}_n$, we denote by $\mathcal{F}_f^{(r)}$ the map from $\mathbb{N}$ to the power set of $\mathcal{B}_n$ as follows
\begin{equation*}
\mathcal{F}_f^{(r)}(k) := \{g\in \mathcal{H}_{n}^{(r)} \cup \{0\} \mid  \mathrm{nl}_{r-1}(f+g)=k\} ,
\end{equation*}
where $\mathcal{H}_{n}^{(r)}$ denotes the set of all $n$-variable homogeneous Boolean functions of degree $r$.
\end{definition}

We need some basic properties of $\mathcal{F}_f^{(r)}$.

\begin{proposition} 
\label{prop:Ff_basic}
Let $f$ be an $n$-variable Boolean function.
\begin{itemize}
\item For any $L \in \mathrm{AGL}(n)$, any $g \in \rmc(r, n)$ where $2\le r\le n$, and for any $k \in \mathbb{N}$,
\[
|\mathcal{F}_f^{(r)}(k)| = |\mathcal{F}_{f \circ L + g}^{(r)}(k)|
\]
Moreover,
\[
\mathcal{F}_{f \circ L + g}^{(r)}(k) = \{ T_r(h \circ L + g) : h \in \mathcal{F}_f^{(r)}(k)\},
\]
where operator $T_r : \mathcal{B}_n \to \mathcal{H}_n^{(r)} \cup \{0\}$ projects a function to its homogeneous degree-$r$ part.
\item For any $b \in \mathbb{F}_2^n$, and for any $k \in \mathbb{N}$,
\[
\mathcal{F}_{f(x)}^{(r)}(k) = \mathcal{F}_{f(x + b)}^{(r)}(k).
\]	
\end{itemize} 
\end{proposition}
\begin{proof}
The first part is a bit trivial. So we skip the proof.

For the second part, let us prove $\mathcal{F}_{f(x)}^{(r)}(k) \subseteq \mathcal{F}_{f(x + b)}^{(r)}(k)$ first. By definition, if $h \in \mathcal{F}_{f(x)}^{(r)}(k)$, then $\nl_{r-1}(f+h) = k$, where $h \in \mathcal{H}_{n}^{(r)} \cup \{0\}$. Since affine transformation preserves $(r-1)$th-order nonlinearity, we have $\nl_{r-1}(f(x+b) + h(x+b)) = k$. Since $h$ is a homogeneous function of degree $r$ or 0, we have $h(x+b) = h(x) + g_2$, where $\deg(g_2) \le r-1$. Thus, 
\begin{eqnarray*}
k & = & \nl_{r-1}(f(x+b) + h(x+b)) \\
& = & \nl_{r-1}(f(x + b) + h(x) + g_2) \\
& = & \nl_{r-1}(f(x + b) + h(x)).
\end{eqnarray*}
 So $h \in \mathcal{F}_{f(x + b)}^{(r)}(k)$.
	
	Replacing $f(x)$ with $f(x + b)$, we have $\mathcal{F}_{f(x + b)}^{(r)}(k) \subseteq \mathcal{F}_{f(x + b + b)}^{(r)}(k) = \mathcal{F}_{f(x)}^{(r)}(k)$. Thus, $\mathcal{F}_{f(x)}^{(r)}(k) = \mathcal{F}_{f(x + b)}^{(r)}(k)$.
\end{proof}

\section{Covering Radius of $\rmc(3, 7)$}

According to the classification of $\rmc(n-1,n-1)/\rmc(r,n-1)$, we classify all $n$-variable Boolean functions into different types to study their $r$th-order nonlinearity.

\begin{definition}
\label{defoftype} Let $fn_0 + \rmc(r, n-1), fn_1 + \rmc(r, n-1), \ldots, fn_{\ell-1}+ \rmc(r, n-1)$ enumerate all non-equivalent cosets in the quotient space  $\rmc(n-1,n-1)/\rmc(r,n-1)$. Let $f = f_1 \| f_2$ be an $n$-variable Boolean function. If $f_1+ \rmc(r, n-1)$ is affine equivalent to $fn_i + \rmc(r, n-1)$, and $f_2+ \rmc(r, n-1)$ is affine equivalent to $fn_j + \rmc(r, n-1)$, where $0\le i,j \le \ell-1$, say $f$ is of \emph{type $(i, j)$} when $i \le j$, of \emph{type $(j, i)$} when $i > j$.
\end{definition}


In this way, we classify all $n$-variable Boolean functions into $\frac{\ell\cdot (\ell+1)}{2}$ types, assuming there are $\ell$ non-equivalent cosets in the quotient space $\rmc(n-1,n-1)/\rmc(r,n-1)$. Theorem \ref{thm:rmc3666_class} says the quotient space $\rmc(6,6) / \rmc(3, 6)$ contains only 11 non-equivalent cosets. Thus we classify all 7-variable Boolean functions into $66$ types accordingly.

The following lemma gives a necessary and sufficient condition for an $n$-variable Boolean function to achieve $r$th-order nonlinearity $t$. In the following sequel, we will use the necessary direction, that is, if the condition \eqref{equ:Ff1} is not satisfied for \emph{some} $h$, then $\nl_{r}(f) < t$.

\begin{lemma} \label{lem:covering_condition}
Let $2\le r\le n-1$.
Let $f = f_1 \| f_2$ be an $n$-variable Boolean function. For any $t \in \mathbb{N}$, $\nl_{r}(f) \ge t$ if and only if
\begin{equation}
\label{equ:Ff1}
\mathcal{F}_{f_1}^{(r)}(h)\subseteq 	\bigcup_{k\geq t-h}\mathcal{F}_{f_2}^{(r)}(k)
\end{equation}
for any $h$.
\end{lemma}
\begin{proof}
For the ``only if'' direction, assume for contradiction that \eqref{equ:Ff1} is not satisfied, we shall prove $\mathrm{nl}_{r}(f) < t$. If \eqref{equ:Ff1} is not satisfied, there exists $g \in \mathcal{H}_{n-1}^{(r)} \cup \{0\}$ such that 
\[
g \in \mathcal{F}_{f_1}^{(r)}(h) \setminus \bigcup_{k\geq t-h}\mathcal{F}_{f_2}^{(r)}(k).
\]
By the definition of $\mathcal{F}_{f_1}^{(r)}$ and $\mathcal{F}_{f_2}^{(r)}$, $\mathrm{nl}_{r-1}(f_1 + g) = h$ and $\mathrm{nl}_{r-1}(f_2 + g) < t-h$. In other words, there exist $(n-1)$-variable Boolean functions $q_1, q_2$ with degree at most $r-1$ such that $\mathrm{d}(f_1+g, q_1) = h$ and $\mathrm{d}(f_2+g,q_2) < t-h$. Observe that $\mathrm{d}(f_1, g+q_1) = \mathrm{d}(f_1+g, q_1) = h$ and $\mathrm{d}(f_2, g+q_2) = \mathrm{d}(f_2+g, q_2) < t-h$. Thus, 
\begin{eqnarray*}
& & \mathrm{d}(f_1 \| f_2, (q_1+g) \| (q_2 +g)) \\
& = & \mathrm{d}(f_1, g+q_1) + \mathrm{d}(f_2, g+q_2) \\
& < & h + (t-h) \\
& = & t.
\end{eqnarray*}
Since $\deg(g) \le r$ and $\deg(q_1), \deg(q_2) \le r-1$, we have $\deg((g+q_1) \| (g+q_2)) \le r$. Therefore, we have found an $n$-variable Boolean function $(q_1+g) \| (q_2 +g)$ of degree at most $r$, whose Hamming distance from $f$ is less than $t$, i.e., $\mathrm{nl}_r(f_1 \| f_2) < t$. Contradiction!

For the ``if '' direction, assuming \eqref{equ:Ff1} is satisfied, let us prove $\mathrm{nl}_r(f) \ge t$. Let $g$ be any $n$-variable Boolean function with degree at most $r$. We shall prove $\mathrm{d}(f, g) \ge t$.

Write $g = g_1 \| g_2$, where $g_1, g_2 \in \mathcal{B}_{n-1}$. Since $\deg(g) \le r$, we have $\deg(g_1), \deg(g_2) \le r$ and $g_1, g_2$ must share the same degree-$r$ terms, denoted by $g_3 \in \mathcal{H}_{n-1}^{(r)} \cup \{0\}$. By the definition of $g_3$, we have $\deg(g_1 + g_3), \deg(g_2 + g_3) \le r-1$.

Let $h = \hdist(f_1, g_1)$. Observe that $\deg(g_1 + g_3) \le r-1$ and $\hdist(f_1 + g_3, g_1 + g_3) = \hdist(f_1, g_1) = h$, which implies that $\nl_{r-1}(f_1 + g_3) \le h$. By the definition of $\mathcal{F}_{f_1}^{(r)}$, we have $g_3 \in \bigcup_{k \le h} \mathcal{F}_{f_1}^{(r)}(k)$. Since \eqref{equ:Ff1} is true, we have
\[
\bigcup_{k \le h} \mathcal{F}_{f_1}^{(r)}(k) \subseteq \bigcup_{k \ge t-h} \mathcal{F}_{f_2}^{(r)}(k).
\] 
Thus, $g_3 \in \bigcup_{k \ge t-h} \mathcal{F}_{f_2}^{(r)}(k)$. By the definition $\mathcal{F}_{f_2}^{(r)}$, $\mathrm{nl}_{r-1}(f_2 + g_3) \ge t - h$. Since $\deg(g_2 + g_3) \le r-1$, we have $\mathrm{d}(f_2 + g_3, g_2 + g_3) \ge \nl_{r-1}(f_2 + g_3) \ge t - h$, that is, $\mathrm{d}(f_2, g_2) \ge t - h$. Therefore, we have $\mathrm{d}(f, g) = \mathrm{d}(f_1, g_1) + \mathrm{d}(f_2, g_2) \ge h + t - h = t$.
\end{proof}

\begin{remark}
\label{re:interchange}
Note that $f_1\|f_2$ is affine equivalent to $f_2\|f_1$. So, $f_1$ and $f_2$ are interchangeable in \eqref{equ:Ff1}.
\end{remark}

The following proposition says both affine transformation and degree-$r$ shift (i.e., adding a function of degree at most r) preserve $r$th-order nonlinearity. The proof is a bit trivial, which is omitted.

\begin{proposition}
\label{prop:nl_invariant}
Let $f = f_1 \| f_2$ be an $n$-variable Boolean function. Then the following operations preserve the $r$th-order nonlinearity:
\begin{itemize}
	\item[(1)] Affine transformation, that is, $f(x) \mapsto f \circ L$ for some $L \in \mathrm{AGL}(n)$. As special cases, it also includes affine transformations on $x_1, \ldots, x_{n-1}$, or $x_n \mapsto x_n + g$, where $g \in \rmc(1, n-1)$.
	\item[(2)] Degree-$r$ shift, that is, $f \mapsto f + g$, where $g \in \rmc(r, n)$. It also includes the special case that, $f_1 \mapsto f_1 + g$ and $f_2 \mapsto f_2 + g$, where $g \in \rmc(r, n-1)$.
\end{itemize}
\end{proposition}

Since affine transformations and degree-$r$ shifts preserve $r$th-order nonlinearity, we can use them to transform the potential functions to a specific form.

\begin{proposition}
\label{prop:canonical_form}
Let $f$ be an $n$-variable Boolean function of type $(i, j)$. Using affine transformations and degree-$r$ shifts, $f$ can be transformed to
\[
fn_i \| (fn_j \circ L + p),
\]
where $L \in \mathrm{AGL}(n-1)$ and $p \in \mathcal{H}_{n-1}^{(r)} \cup \{0\}$.

Moreover, if $fn_j \circ L_1 + \rmc(r, n-1), fn_j \circ L_2 + \rmc(r, n-1), \ldots, fn_j \circ L_\ell + \rmc(r, n-1)$ enumerate the cosets in the $\mathrm{AGL}(n-1)$-orbit of $fn_j + \rmc(r, n-1)$, then we can assume $L \in \{L_1, L_2, \ldots, L_\ell\}$.
\end{proposition}
\begin{proof}
Let $f = f_1 \| f_2$. Since $f$ is of type $(i, j)$, we know $f_1 + \rmc(r, n-1)$ is affine equivalent to $fn_i + \rmc(r, n-1)$, and $f_2 + \rmc(r, n-1)$ is affine equivalent to $fn_j + \rmc(r, n-1)$. Thus, there exist $L_1, L_2 \in \mathrm{AGL}(n-1)$, and $p_1, p_2 \in \rmc(r, n-1)$ such that
\[
f_1 = fn_i \circ L_1 + p_1,
\]
\[
f_2 = fn_j \circ L_2 + p_2.
\]

Applying transformation $L_1^{-1}$ on $(x_1, \ldots, x_{n-1})$, $f_1$ becomes $fn_i + p_1 \circ L_1^{-1}$, and $f_2$ becomes $fn_j \circ L_2 \circ L_1^{-1} + p_2 \circ L_1^{-1}$. Applying degree-$r$ shift $p_1 \circ L_1^{-1}$ to both $f_1$ and $f_2$, then $f_1$ becomes $fn_i$, and $f_2$ becomes $fn_j \circ L_2 \circ L_1^{-1} + p_2 \circ L_1^{-1} + p_1 \circ L_1^{-1}$. Let $L = L_2 \circ L_1^{-1}$ and $p = p_2 \circ L_1^{-1} + p_1 \circ L_1^{-1}$, which proves the first part. (By performing a degree-$r$ shift, we can eliminate the monomials of degree at most $r-1$, and thus $p$ can be homogeneous.)

Assume $fn_j \circ L_1 + \rmc(r, n-1), fn_j \circ L_2 + \rmc(r, n-1), \ldots, fn_j \circ L_\ell + \rmc(r, n-1)$ enumerate the cosets in the $\mathrm{AGL}(n-1)$-orbit of $fn_j + \rmc(r, n-1)$. Then there exists some $L' \in \{L_1, L_2, \ldots, L_\ell\}$ such that $fn_j \circ L' + \rmc(r, n-1) = fn_j \circ L + \rmc(r, n-1)$, that is,
$fn_j \circ L' = fn_j \circ L + p'$ for some $p' \in \rmc(r, n-1)$. Thus $f_2 = fn_j \circ L' + (p + p')$, as desired.
\end{proof}
\begin{remark} Similarly, we can prove if $fn_i \circ L_1 + \rmc(r, n-1), \ldots, fn_i \circ L_\ell + \rmc(r, n-1)$ enumerate the cosets in the $\mathrm{AGL}(n-1)$-orbit of $fn_i + \rmc(r, n-1)$, $f$ can be transformed to
\[
fn_i \| (fn_j \circ L^{-1} + p),
\]
where $L \in \{L_1, L_2, \ldots, L_\ell\}$ and $p \in \mathcal{H}_{n-1}^{(r)} \cup \{0\}$.
\end{remark}

The goal is to prove any 7-variable Boolean function cannot achieve third-order nonlinearity 21. Our overall strategy is
\begin{itemize}
	\item Using the algebraic properties of $fn_i$, $i = 0, 1, \ldots, 10$, we exclude 62 types (among 66). Notably, we  present our results in general forms, so that they be potentially used to prove other upper bounds on $\rho(r, n)$.
	\item For Boolean function $f$ of type (2, 9), (2, 10), or (3, 10), we prove either $\nl_3(f) \le 20$, or $f$ is equivalent to a type (6, 10) function.
	\item Verify type (6,10) Boolean functions cannot achieve third-order nonlinearity 21 using Lemma \ref{lem:covering_condition}. We first transform the function into a specific form, which significantly reduces the search space. Then, with the assistance of computers, we verify that the condition \eqref{lem:covering_condition} (in Lemma \ref{lem:covering_condition}) does not hold for $h = 6$.
\end{itemize}

In two places, we rely on computers to complete our proof of $\rho(3, 7) \le 20$. Beyond $\rho(3, 7)$, we also propose a general method for upper bounding $\rho(r, n)$ by $\max_{g \in \mathcal{B}_{n-1}} \{ \nl_r(g) + \nl_{r-1}(g)\}$, which improves the well-known upper bound $\rho(r, n) \le \rho(r,n-1) + \rho(r-1, n-1)$.


In terms of techniques, Wang's proof of $\crrm(2,7) = 40$ \cite{Wang19} is to compare the \emph{cardinality} of the left hand side and right hand side of \eqref{equ:Ff1} (or their intersection with a specific set), and argue condition \eqref{equ:Ff1} cannot hold. For proving $\crrm(3, 7) = 20$, in addition to the cardinality argument, we apply affine transformations to transform interested functions to a specific type, and dig deep into the set $\mathcal{F}_{fn_i}^{(3)}(k)$ (with the assistance of computers) to prove \eqref{equ:Ff1} cannot hold.

\subsection{Exclude 62 types}

Let $f$ be an $n$-variable Boolean function, and let $1 \le r \le n-1$.
For notational convenience, define
\[
\ml_r(f) := \max_{g \in \mathcal{H}^{(r+1)}_n \cup \{0\}} \nl_{r}(f+g).
\]
In other words, $\ml_r(f)$ is the maximum $r$th-order nonlinearity of a function whose degree $>r+1$ parts cocincides with $f$.

The following proposition is easy to prove, and the proof is omitted.

\begin{proposition}
\label{prop:ml_properties}
Let $f$ be an $n$-variable Boolean function, and let $1 \le r \le n-1$.
\begin{itemize}
\item Let $g$ be any $n$-variable Boolean function of degree at most $r+1$. Then $\ml_r(f) = \ml_r(f + g)$.
\item Let $L \in \mathrm{AGL}(n)$ be any affine transformation. Then $\ml_r(f) = \ml_r(f \circ L)$.
\end{itemize}
\end{proposition}

\begin{lemma}
\label{lem:nlofone} Let $fn_0 + \rmc(r, n-1), fn_1 + \rmc(r, n-1), \ldots$ enumerate all non-equivalent cosets of the quotient space $\rmc(n-1, n-1) / \rmc(r, n-1)$. Let $f = f_1 \| f_2$ be an $n$-variable Boolean function of \emph{type $(i,j)$}, that is, $f_1 + \rmc(r, n-1)$ is affine equivalent to $fn_i+\rmc(r, n-1)$ and $f_2 + \rmc(r, n-1)$ is affine equivalent to $fn_j + \rmc(r, n-1)$. Then
 \[
 \nl_r(f) \le \min\left\{ \nl_r(fn_i) + \ml_{r-1}(fn_j), \nl_r(fn_j) + \ml_{r-1}(fn_i)\right\}.
 \]
\end{lemma}
\begin{proof}
    Let $g_1\in \rmc(r,n-1)$ be such that $\hdist(f_1,g_1) =\nl_r(f_1)$. By the definition of $\ml_{r-1}(f_2)$, we know $\nl_{r-1}(f_2 + g_1)\leq \ml_{r-1}(f_2)$, which implies that, there exists $g_2 \in \rmc(r-1,n-1)$ such that $\hdist(f_2+g_1, g_2)\leq \ml_{r-1}(f_2)$. Letting $g=g_1\|(g_1 + g_2) \in \rmc(r,n)$, we have $\hdist(f,g) = \hdist(f_1, g_1) + \hdist(f_2, g_1+g_2) = \hdist(f_1, g_1) + \hdist(f_2 + g_1, g_2) \le \nl_r(f_1)+ \ml_{r-1}(f_2)$. Hence, $\nl_r(f)\le \nl_r(f_1)+ \ml_{r-1}(f_2)$.
    
    By our condition,
    \begin{eqnarray*}
    f_1 & = & fn_i \circ L_1 + h_1,\\
    f_2 & = & fn_j \circ L_2 + h_2,
    \end{eqnarray*}
    where $L_1, L_2 \in \agl(n)$ and $\deg(h_1), \deg(h_2) \le r$. By Proposition \ref{prop:nl_invariant} and Proposition \ref{prop:ml_properties}, we have $\nl_r(f_1) = \nl_r(fn_i)$ and $\ml_{r-1}(f_2) = \ml_{r-1}(fn_j)$. So, $\nl_r(f)\le \nl_r(fn_i)+ \ml_{r-1}(fn_j)$. Inequality $\nl_r(f)\le \nl_r(fn_j)+ \ml_{r-1}(fn_i)$ is similar to prove.
\end{proof}

\begin{lemma}
	\label{lem:nl3oddeven}
 Let $1 \le r \le n-2$.
	Let $f = f_1 \| f_2$ be an $n$-variable Boolean function. If $\nl_r(f)$ is odd, then one of $\nl_r(f_1), \nl_r(f_2)$ is odd, and the other is even.
\end{lemma}
\begin{proof}
	Observe that if the weight of a Boolean function is odd, its degree is maximal.
	$\nl_r(f) $ being odd implies that there exists $g \in \rmc(r, n)$ such that $\hdist(f, g) $ is odd, that is, $\wt(f+g) $ is odd. Thus $\deg(f+g) = n$, which implies that $\deg(f) = n$.
	
	Notice that $f = (f_1 + f_2)x_n + f_1$, where $f_1, f_2 \in \mathcal{B}_{n-1}$. Since $\deg(f) = n$, we have $\deg(f_1 + f_2) = n-1$, that is, $f_1 + f_2$ contains the highest-degree term $x_1x_2 \cdots x_{n-1}$. Thus, exactly one of $\deg(f_1)$ and $\deg(f_2)$ is $n-1$, which implies that exactly one of $\nl_r(f_1),\nl_r(f_2)$ is odd.
\end{proof}

From the previous lemma, we know 7-variable functions of type $(i, i)$, $0 \le i \le 10$, cannot have third-order nonlinearity 21.

Now we are ready to exclude 62 types (among 66 in total).

\begin{corollary} \label{lem:li_theorem} 
Let $f$ be a 7-variable Boolean function of type $(i, j)$, $0 \le i < j \le 10$. If $(i, j) \not\in \{(2,9), (2, 10), (3, 10), (6,10)\}$, then $\nl_3(f) \le 20$.	
\end{corollary}
\begin{proof}
	By Lemma \ref{lem:nlofone} and Table \ref{table:fn10_nl3} in Appendix \ref{appendix_properties_fni}, we can clearly deduce that 7-variable Boolean functions of only four types which are type $(i, j) \in \{(2,9), (2, 10), (3, 10), (6,10)\}$ can probably achieve the third-order nonlinearity 21. 
\end{proof}


\begin{theorem}
\label{thm:rhorn_atmost} Let $fn_0 + \rmc(r, n-1), fn_1 + \rmc(r, n-1), \ldots$ enumerate all non-equivalent cosets in the quotient space $\rmc(n-1, n-1) / \rmc(r, n-1)$. Then
\[
\rho(r,n) \le \max_k \left\{ \nl_r(fn_k) + \ml_{r-1}(fn_k) \right\}.
\]
\end{theorem}
\begin{proof} Let $f$ be an $n$-variable Boolean function with maximum $r$th-order nonlinearity, i.e., $\nl_r(f) = \rho(r, n)$. We shall prove $\nl_r(f) \le \max_{g \in \mathcal{B}_{n-1}}\{\nl_r(g)+\ml_{r-1}(g)\}$.

Suppose $f$ is of type $(i, j)$. By Lemma \ref{lem:nlofone}, 
\begin{eqnarray*}
 \nl_r(f) & \le & \min\left\{ \nl_r(fn_i) + \ml_{r-1}(fn_j), \nl_r(fn_j) + \ml_{r-1}(fn_i)\right\} \\
 & \le &  \max\left\{ \nl_r(fn_i) + \ml_{r-1}(fn_i), \nl_r(fn_j) + \ml_{r-1}(fn_j)\right\} \\
 & \le & \max_k \left\{ \nl_r(fn_k) + \ml_{r-1}(fn_k) \right\},
\end{eqnarray*}
where the second step is true because $\min\{a+b, c+d\} \le \max\{a+c, b+d\}$ holds for any $a, b, c, d \in \mathbb{N}$.
\end{proof}

\begin{corollary}
\label{cor:rhorn_atmost}
 $\rho(r,n) \le \max_{g \in \mathcal{B}_{n-1}}\{\nl_r(g)+\nl_{r-1}(g)\}$.
\end{corollary}
\begin{proof}
Let $fn_0 + \rmc(r, n-1), fn_1 + \rmc(r, n-1), \ldots$ enumerate all non-equivalent cosets in the quotient space $\rmc(n-1, n-1) / \rmc(r, n-1)$. By Theorem \ref{thm:rhorn_atmost}, 
\[
 \nl_r(f) \le \max_k \{ \nl_r(fn_k) + \ml_{r-1}(fn_k) \}.
\]

We will prove $\max_k \{ \nl_r(fn_k) + \ml_{r-1}(fn_k) \} \le \max_{g \in \mathcal{B}_{n-1}}\{\nl_r(g)+\nl_{r-1}(g)\}$. It suffices to prove, for any $k$, $\nl_r(fn_k) + \ml_{r-1}(fn_k) \le \max_{g \in \mathcal{B}_{n-1}}\{\nl_r(g)+\nl_{r-1}(g)\}$. Recall the definition of $\ml_{r-1}(fn_k)$, let $h \in \rmc(r, n-1)$ such that $\nl_{r-1}(fn_k + h) = \ml_{r-1}(fn_k)$. By Proposition \ref{prop:nl_invariant}, $\nl_r(fn_k + h) = \nl_r(fn_k)$.
So we have
\begin{eqnarray*}
& & \nl_r(fn_k) + \ml_{r-1}(fn_k) \\
& = & \nl_r(fn_k + h) + \nl_{r-1}(fn_k + h) \\
& \le & \max_{g \in \mathcal{B}_{n-1}}\{\nl_r(g)+\nl_{r-1}(g)\},
\end{eqnarray*}
as desired.
\end{proof}

Compared with the upper bound $\rho(r, n) \le \rho(r-1, n-1) + \rho(r, n-1)$ \cite{Schatz81}, our bound $\rho(r,n) \le \max_{g \in \mathcal{B}_{n-1}}\{\nl_r(g)+\nl_{r-1}(g)\}$ is better, since $\nl_r(g) \le \rho(r, n-1)$ and $\nl_{r-1}(g) \le \rho(r-1, n-1)$ always hold for any $g \in \mathcal{B}_{n-1}$. Most likely, the equalities cannot hold \emph{simultaneously}.

For example, from recursion $\rho(r, n) \le \rho(r-1, n-1) + \rho(r, n-1)$, we have $\rho(3, 7) \le \rho(2, 6) + \rho(3, 6) = 18 + 8 = 26$, which is much weaker than the following.

\begin{corollary}
    $\rho(3, 7) \le 22.$
\end{corollary}
\begin{proof}
Immediate from Theorem \ref{thm:rhorn_atmost}  and Table \ref{table:fn10_nl3} in Appendix \ref{appendix_properties_fni}.
\end{proof}

\subsection{Reduce type (2,9), (2, 10), (3,10) to (6,10)}

Now we are left with type (2,9), (2, 10), (3, 10) and (6, 10). For function $f$ of type (2, 9), (2, 10), or (3, 10), we will prove $\nl_3(f) \le 20$ or $f$ is affine equivalent to a type (6, 10) function.

\begin{lemma}
\label{lem:type210}
Let $f$ be a 7-variable Boolean function of type (2, 10). Then $\nl_3(f) \le 20$ or $f$ is affine equivalent to a type (6, 10) Boolean function.	
\end{lemma}
\begin{proof}
By Proposition \ref{prop:canonical_form}, $f$ can be transformed to $(fn_2 \circ L + g) \| fn_{10}$, where $L \in \mathrm{AGL}(6)$ and $g \in \rmc(3, 6)$. For convenience, let $f_1 = fn_2 \circ L + g$, $f_2 = fn_{10}$. Observe that
\[
fn_{10} = x_1 x_2 x_3 x_4 x_5 x_6 + fn_3,
\]
where $fn_3$ is a homogeneous function of degree 4.
$fn_2$ is also a homogeneous function of degree 4. Since $fn_2 + \rmc(3, 6)$ and $fn_3 + \rmc(3, 6)$ are not affine equivalent, we have $fn_2 \circ L + \rmc(3,6) \not= fn_3 + \rmc(3, 6)$. So we have
\begin{eqnarray*}
f_1 + f_2 & = & fn_2 \circ L + g + fn_{10} \\
& = & x_1 x_2 x_3 x_4 x_5 x_6 + fn_2 \circ L + fn_3 + g \\
& = & x_1 x_2 x_3 x_4 x_5 x_6 + h_4 + g_3,
\end{eqnarray*}
where $h_4 \in \mathcal{H}_6^{(4)}$ is \emph{nonzero} homogeneous function of degree 4, and $\deg(g_3) \le 3$.

We claim there exists $k \in \{1, 2, \ldots, 6\}$, such that $(f_1 + f_2)x_k$ contains both the degree-6 monomial and degree-5 monomial(s). Without loss of generality, assume $h_4$ contains the monomial $x_1x_2x_3x_4$. Letting $k = 5$, we can see $(f_1 + f_2)x_k$ contains both the degree-6 monomial and degree-5 monomial $x_1x_2x_3x_4x_5$.

Applying affine transformation $x_7 \mapsto x_7 + x_k$, $f_1 \| f_2$ becomes 
\begin{eqnarray*}
f' & = & (f_1 + f_2)(x_7 + x_k) + f_1 \\
   & = & (f_1 + f_2)x_7 + f_1 + (f_1 + f_2)x_k \\
   & = & (f_1 + (f_1 + f_2)x_k) \| (f_2 + (f_1 + f_2)x_k).
\end{eqnarray*}
For convenience, let $f_1' = f_1 + (f_1 + f_2)x_k$ and $f_2' = f_2 + (f_1 + f_2)x_k$, so $f' = f_1' \| f_2'$. We have shown that $(f_1 + f_2)x_k$ contains the degree-6 monomial and \emph{at least} one degree-5 monomial. Combining with the fact that $\deg(f_1) = 4$ and $\deg(f_2) = 6$, we have $\deg(f_1') = 6$ and $\deg(f'_2) = 5$. From Table \ref{table:fn10_nl3}, we know $f'$ is of type $(i, j)$, where $i \in \{4,5,6\}$ and $j \in \{7,8,9,10\}$.

If $(i, j) \not= (6,10)$, we have $\nl_3(f) = \nl_3(f') \le 20$ by Lemma \ref{lem:nlofone}. Otherwise, we have transformed $f$ into a function of type (6, 10).
\end{proof}

The following lemma is proved with the assistance of a computer.

\begin{lemma}
\label{lem:fn29_310_exceptions}
Let 7-variable Boolean function $f = fn_i \| (fn_j + g)$, where $g \in \mathcal{H}_6^{(3)} \cup \{0\}$, $(i, j) \in \{ (2, 9), (3, 10)\}$. Then $\nl_3(f) \le 20$.
\end{lemma}
\begin{proof} By Lemma \ref{lem:covering_condition}, if
\[
\mathcal{F}_{fn_i}^{(3)}(l) \not\subseteq \bigcup_{k \ge 21 - l} \mathcal{F}_{fn_j + g}^{(3)}(k) = \bigcup_{k \ge 21 - l} \mathcal{F}_{fn_j}(k) + g
\]
for \emph{some} $l$, then $\nl_3(f) \le 20$.

\textbf{Case 1. $(i, j) = (2, 9)$.} From Table \ref{table:fn10_nl3}, we know $\nl_2(fn_2) = 6$, which implies that $0 \in \mathcal{F}_{fn_2}^{(3)}(6)$. Since $0 \in \mathcal{F}_{fn_2}^{(3)}(6)$, it suffices to enumerate all $g \in \mathcal{F}_{fn_9}^{(r)}(15)$. Otherwise $\mathcal{F}_{fn_2}^{(3)}(6) \subseteq \bigcup_{k \ge 15} \mathcal{F}_{fn_9+g}^{(3)}(k)$ cannot hold. We use a computer to enumerate $g \in \mathcal{F}_{fn_9}^{(3)}(15)$, and verify 
\[
\mathcal{F}_{fn_2}^{(3)}(8) \subseteq \mathcal{F}_{fn_9}^{(3)}(13) \cup \mathcal{F}_{fn_9}^{(3)}(15) + g
\]
does not hold for any $g$.

\textbf{Case 2. $(i, j) = (3, 10)$.} From Table  \ref{table:fn10_nl3}, we know $\nl_2(fn_{10}) = 9$, which implies that $0 \in \mathcal{F}_{fn_{10}}^{(3)}(9)$. Since $0 \in \mathcal{F}_{fn_{10}}^{(3)}(9)$, it suffices to enumerate $g \in \mathcal{F}_{fn_3}^{(3)}(12)\cup \mathcal{F}_{fn_3}^{(3)}(14)$. We use a computer to enumerate $g$, and it turns out there are 6912 $g$ satisfying $\mathcal{F}_{fn_{10}}^{(3)}(7) \subseteq \mathcal{F}_{fn_3}^{(3)}(14) + g$. However, among those 6912 $g$, 
\[
\mathcal{F}_{fn_{10}}^{(3)}(9) \subseteq \mathcal{F}_{fn_3}^{(3)}(12) \cup \mathcal{F}_{fn_3}^{(3)}(14) + g
\]
does not hold for any $g$.
\end{proof}

Using a personal computer with 1.4 GHz Intel Core i5 and 16 GB RAM, it takes 7208 seconds (2 hours and 8 seconds) to precompute $\mathcal{F}_{fn_i}^{(3)}$. For the $fn_2 \| (fn_9 + g)$ case, we enumerate $|\mathcal{F}_{fn_9}^{(3)}(15)| = 5760$ $g$; a typical run of the verification takes 40 seconds.  For the $fn_3 \| (fn_{10} + g)$ case, we numerate $|\mathcal{F}_{fn_3}^{(3)}(12)\cup \mathcal{F}_{fn_3}^{(3)}(14)| = 974592$ $g$ in the first round, and enumerate 6912 $g$ in the second round, which takes 90 seconds combined.

\begin{lemma}
\label{lem:type29_310}
Let $f$ be a 7-variable Boolean function of type $(i, j)$, where $(i, j) \in \{(2,9), (3, 10)\}$. Then $\nl_3(f) \le 20$, or $f$ is affine equivalent to a function of type $(6, 10)$.
\end{lemma}
\begin{proof}
	By Proposition \ref{prop:canonical_form}, $f$ can be transformed to
	\[
	f_1 \| f_2 = fn_i \| (fn_j \circ L + p),
	\]
	where $L \in \mathrm{AGL}(6)$, $p \in \mathcal{H}_6^{(3)} \cup \{0\}$.
	
	\textbf{Case 1. $fn_j \circ L + \rmc(3, 6) = fn_j + \rmc(3, 6)$.}  Then $f_1 \| f_2 = fn_i \| (fn_j + g)$ for some $g \in \rmc(3, 6)$. By Lemma \ref{lem:fn29_310_exceptions}, we have $\nl_3(f) = \nl_3(f_1 \| f_2) \le 20$.
	
	\textbf{Case 2. $fn_j \circ L + \rmc(3, 6) \not= fn_j + \rmc(3, 6)$.} Observe that 
	\[
	fn_i = x_1x_2x_3x_4x_5x_6 + fn_j,
	\] where $fn_j$ is a homogeneous function of degree 4.
	So we have
	\begin{eqnarray*}
		f_1 + f_2 & = & fn_i + fn_j \circ L + p \\
		& = & x_1x_2x_3x_4x_5x_6 + fn_j + fn_j \circ L + p \\
		& = & x_1x_2x_3x_4x_5x_6 + h_4 + g_3,
	\end{eqnarray*}
	where $h_4 \in \mathcal{H}_6^{(4)}$ is a \emph{nonzero} homogeneous function of degree 4, and $g_3 \in \rmc(3, 6)$. The rest of the proof is the same as Lemma \ref{lem:type210}. That is, we find some $k$ such that $(f_1 + f_2)x_k$ contains both the degree-6 monomial and at least one degree-5 monomial. Applying affine transformation $x_7 \mapsto x_7 + x_k$, we transform $f_1 \| f_2$ to $f' = f_1' \| f_2'$, where $\deg(f_1') = 5$ and $\deg(f_2') = 6$, from which we can conclude $\nl_3(f) = \nl_3(f') \le 20$ or $f'$ is of type $(6, 10)$.
\end{proof}

\subsection{Type (6,10)}

The following lemma is proved with the assistance of computers.

\begin{lemma} There exist affine transformations $L_1, L_2, \ldots, L_\ell \in \mathrm{AGL}(6)$, $\ell = 888832$, such that
\begin{itemize}
	\item $fn_{10} \circ L_1 + \rmc(3, 6), \ldots, fn_{10} \circ L_\ell + \rmc(3, 6)$ enumerate the cosets in the $\mathrm{AGL}(6)$-orbit of $fn_{10} + \rmc(3, 6)$ in $\rmc(6,6)/\rmc(3,6)$.
	\item For any $L \in \{L_1, L_2, \ldots, L_\ell\}$, for any $g \in \mathcal{H}_6^{(3)} \cup \{0\}$,
 \begin{equation}
 \label{equ:type610_noncovering}	
 \mathcal{F}_{fn_6 \circ L^{-1} + g}^{(3)}(6) \not \subseteq \mathcal{F}_{fn_{10}}^{(3)}(15).
 \end{equation}
\end{itemize}
\end{lemma}
\begin{proof}
Firstly, we use a breadth-first search (BFS) algorithm to generate $L_1, L_2, \ldots, L_\ell \in \mathrm{AGL}(6)$ so that $fn_{10} \circ L_1 + \rmc(3, 6), \ldots, fn_{10} \circ L_\ell + \rmc(3, 6)$ enumerate the cosets in the $\mathrm{AGL}(6)$-orbit of $fn_{10} + \rmc(3, 6)$.

Let $L_i(x) = A_i(x) + b_i$. By Proposition \ref{prop:Ff_basic},
\begin{eqnarray*}
\mathcal{F}_{fn_6 \circ L_i^{-1} + g}^{(3)}(6)
& = & \mathcal{F}_{fn_6 \circ L_i^{-1}}^{(3)}(6) + g \\
& = & \mathcal{F}_{fn_6(A_i^{-1}x + A_i^{-1}b_i)}^{(3)}(6) + g \\
& = & \mathcal{F}_{fn_6(A_i^{-1}x)}^{(3)}(6) + g. \\
\end{eqnarray*}
So $b_i$ is inconsequential. We keep those distinct $A_i$, denoted by $\mathcal{A} = \{A_1, A_2, \ldots, A_\ell\}$. Using the generators of $\mathrm{AGL}(6)$ in Section \ref{appendix_generators}, it turns our that $|\mathcal{A}| = 130843$, much smaller than $\ell = 888832$.

\begin{algorithm}
	\caption{Compute $\mathcal{A}$ using BFS}\label{alg:bfs}
	\begin{algorithmic}
		\State /* $G_1, G_2 $ are generators of the affine group */
		\State $\mathcal{G} \gets \{G_1, G_2 \}$
		\State /* $Q$ is a queue, $I$ is the identity transformation */
		\State $Q \gets [(fn_{10} + \rmc(3,6), I)]$
		\State $\mathcal{A} \gets \{\}$
		\State $visited \gets \{\}$
		\While{$Q$ is not empty}
		\State $(f, L) \gets Q.pop()$
		\For{$G \in \mathcal{G}$}
		\If{$f \circ G + \rmc(3,6) \not\in$ visited}
		\State $visited \gets visited \cup \{f \circ G + \rmc(3,6) \}$
		\State $Q.push((f \circ G + \rmc(3,6), L \circ G))$
		\State $\mathcal{A} \gets \mathcal{A} \cup \{(L \circ G).A\}$
		\EndIf
		\EndFor
		\EndWhile
		\State return $\mathcal{A}$
	\end{algorithmic}
\end{algorithm}

For each $A \in \mathcal{A}$, for each $g \in \mathcal{H}_6^{(3)} \cup \{0\}$, we verify
\[
 \mathcal{F}_{fn_6(A^{-1}x)}^{(3)}(6) + g \not \subseteq \mathcal{F}_{fn_{10}}^{(3)}(15).
\]
Note that $|\mathcal{F}_{fn_6(x)}^{(3)}(6)| = 32$. Using Proposition \ref{prop:Ff_basic}, we can compute $\mathcal{F}_{fn_6(A^{-1}x)}^{(3)}(6)$ by applying transformation $x \mapsto A^{-1}x$ to each function in $\mathcal{F}_{fn_6(x)}^{(3)}(6)$ and trim all monomials with degree less than 3.

\begin{algorithm}
\caption{Verification algorithm for type (6,10)}\label{alg:type610_main}
\begin{algorithmic}
\State /* Assume $\mathcal{A}$ is precomputed */
\For{$A \in \mathcal{A}$}
\State /* Compute $\mathcal{F}_{fn_6(A^{-1}x)}^{(3)}(6)$ */
\State Let $S = \{ f(A^{-1}x) \mod \rmc(3,6) : f \in \mathcal{F}_{fn_6}^{(3)}(6)\}$.
\State Denote by $S = \{s_0, s_1, \ldots, s_{31}\}$.
\State /* Check if there exists $g$ such that $S+g$ is a subset of $\mathcal{F}_{fn_{10}}^{(3)}(15)$*/
\For{$t \in \mathcal{F}_{fn_{10}}^{(3)}(15)$}
\State $g \gets s_0 + t$
\State issubset $\gets$ true
\For{$i = 1,2,\ldots, 31$}
\If{$g + s_i \not\in \mathcal{F}_{fn_{10}}^{(3)}(15)$}
\State issubset $\gets$ false
\EndIf
\EndFor
\If{issubset = true}
\State return fail
\EndIf
\EndFor
\EndFor
\State return succeed
\end{algorithmic}
\end{algorithm}

\end{proof}

In theory, the BFS algorithm (Algorithm \ref{alg:bfs}) requires $O(\ell)$ basic operations on $n$-variable Boolean functions, where $\ell$ is the orbit length of $fn_{10} + \rmc(3,6)$. Each basic operation on an $n$-variable Boolean function requires $O(n2^n)$ basic arithmetic operations, including conversion between ANF and truth table, affine transformation, etc. So the total running time is $O(n2^n \ell)$.

Using a personal computer with 1.4 GHz Intel Core i5 and 16 GB RAM, it takes 80 seconds to generate the orbit of $fn_{10} + \rmc(3,6)$, that is, to precompute set $\mathcal{A}$.

The main verification algorithm for type (6,10) functions is given in Algorithm \ref{alg:type610_main}. The algorithm requires $|\mathcal{A}|\times |\mathcal{F}_{fn_{10}}^{(3)}(15)|\times|\mathcal{F}_{fn_{6}}^{(3)}(6)| \approx 1.447 * 10^{11}$ basic 6-variable Boolean function operations; each basic operation on $n$-variable Boolean functions requires $O(n2^n)$ basic arithmetic operations. Using a personal computer with 1.4 GHz Intel Core i5 and 16 GB RAM, a typical run takes about 8 hours and 30 minutes.

\begin{theorem}
\label{thm:type610} Let $f$ be a 7-variable Boolean function of type $(6, 10)$. Then $\nl_3(f) \le 20$.
\end{theorem}
\begin{proof} Since $f$ is of type $(6, 10)$, by Proposition \ref{prop:canonical_form}, while preserving its third-order nonlinearity, $f$ can be transformed to
\[
(fn_6 \circ L^{-1} + g) \| fn_{10},
\]
where $L \in \{L_1, L_2, \ldots, L_\ell\}$ and $g \in \mathcal{H}_6^{(3)} \cup \{0\}$. Since \eqref{equ:type610_noncovering} always holds, by Lemma \ref{lem:covering_condition}, we have
\[
\nl_3(f) = \nl_3((fn_6 \circ L^{-1} + g) \| fn_{10}) < 21.
\]	
\end{proof}

\subsection{Putting together}

Now we are ready to prove our main result.

\begin{theorem}
	$\crrm(3,7) = 20$.
\end{theorem}
\begin{proof}
Let us prove $\crrm(3, 7) \le 20$ first. Because changing the value of one point results in changing the third-order nonlinearity by at most 1, it suffices to prove 7-variable Boolean function cannot have third-order nonlinearity 21. Let $f$ be a 7-variable Boolean function of type $(i, j)$. By Lemma \ref{lem:nl3oddeven}, we can assume $i \not= j$.

If $(i, j) \not\in \{(2,9), (2,10), (3, 10), (6, 10)\}$, by Corollary \ref{lem:li_theorem}, we have $\nl_3(f) \le 20$. If $(i, j) = (6, 10)$, by Theorem \ref{thm:type610}, we have $\nl_3(f) \le 20$. If $(i, j) \in \{(2,9), (2,10), (3, 10)\}$, by Lemma \ref{lem:type210}, Lemma \ref{lem:type29_310}, and Theorem \ref{thm:type610}, we have $\nl_3(f) \le 20$. 
So we conclude $\crrm(3, 7) \le 20$.

On the other hand, Hou proves that there exist 7-variable Boolean functions $f$ with $\nl_3(f) = 20$ for the first time \cite{Hou93}. The authors in \cite{WTP18} characterize all functions in $\rmc(4, 7)$ with third-order nonlinearity 20. For example,
\[
x_1x_2x_3x_4+x_1x_4x_6x_7+x_2x_3x_6x_7+x_3x_4x_5x_7
\]
has third-order nonlinearity 20.
\end{proof}

\section{Other upper bounds on $\crrm(r, n)$}
 
 Using inequality  $\crrm(r, n) \le \crrm(r-1, n-1) + \crrm(r, n-1)$, we can prove the following upper bounds on $\crrm(3, n)$ and $\crrm(4, n)$ for $n = 8, 9, 10$.
 
 \begin{corollary} $\crrm(3, 8) \le 60$, $\crrm(3, 9) \le 156$, 
$\crrm(3, 10) \le 372$, $\crrm(4, 8) \le 28$, $\crrm(4, 9) \le 88$, 
$\crrm(4, 10) \le 244$.
 \end{corollary}
 \begin{proof}
 $
 \crrm(3, 8) \le \crrm(2,7) + \crrm(3, 7) = 40 + 20 = 60,
 $
 where $\crrm(2,7) = 40$ is proved in \cite{Wang19}.
 
  $
 \crrm(3, 9) \le \crrm(2,8) + \crrm(3, 8) = 96 + 60 = 156,
 $
 where $\crrm(2,8) \le 96$ is proved in \cite{Wang19}.
 
   $
 \crrm(3, 10) \le \crrm(2,9) + \crrm(3, 9) = 216 + 156 = 372,
 $
 where $\crrm(2,9) \le 216$ is proved in \cite{Wang19}.
 
 $
 \crrm(4, 8) \le \crrm(3,7) + \crrm(4, 7) = 20 + 8 = 28,
 $
 where $\crrm(4,7) = 8$ is proved in \cite{Mcl79}.
 
 $
 \crrm(4, 9) \le \crrm(3,8) + \crrm(4, 8) = 60 + 28 = 88.
 $
 
 $
 \crrm(4, 10) \le \crrm(3,9) + \crrm(4, 9) = 156 + 88 = 244.
 $
 \end{proof}
 
\section{Generators of the affine transformation group}
\label{appendix_generators}

Waterhouse \cite{Wat89} proves that the general linear group over finite fields can be generated by two elements, which are explicitly described. Based on Waterhouse's results, we will prove the affine transformation group over any finite field can be generated by two elements. Let $\alpha$ denote a generator of the multiplicative group of $\mathbb{F}_q$.

\begin{theorem} \cite{Wat89}
\label{thm:two_gen}
	For $n \ge 3$, the group $\mathrm{GL}(n) = \mathrm{GL}(n, \mathbb{F}_q)$ is generated by two elements $A=I + E_{n1}+(\alpha-1)E_{22}$ and $B = E_{12} + E_{23} + \ldots + E_{n1}$, where $E_{ij}$ denotes the matrix with entry 1 in row $i$ and column $j$, and $0$ elsewhere.
\end{theorem}

\begin{theorem} \cite{Wat89}
	\label{thm1:two_gen}
	For $n = 2$ and $q\ge 3$, the group $\mathrm{GL}(n) = \mathrm{GL}(2, \mathbb{F}_q)$ is generated by two elements  $A= \begin{pmatrix} 0 & -\alpha \\ 1 & \beta+\beta^{q} \end{pmatrix}$ and $B= \begin{pmatrix} \alpha & 0 \\ 0 & 1 \end{pmatrix}$, where $\beta$ denotes a generator of the multiplicative group over a quadratic extension of $\mathbb{F}_q$.
\end{theorem}

As for $q=2$, the group $\mathrm{GL}(n) = \mathrm{GL}(2, \mathbb{F}_2)$ can be generated by two elements $\begin{pmatrix} 0 & 1 \\ 1 & 1 \end{pmatrix}$ and $\begin{pmatrix} 1 & 1 \\ 0 & 1 \end{pmatrix}$ \cite{Wat89}.

Recall that the affine transformation group $\mathrm{AGL}(n,\mathbb{F}_q)$ consists of all elements of the form
$
(A, b) \in \mathrm{GL}(n,\mathbb{F}_q) \times \mathbb{F}_q^n,
$
which corresponds to the affine transformation $x \mapsto Ax + b$. The group operation is defined by $(A_1, b_1) \circ (A_2, b_2) = (A_1 A_2, A_1 b_2 + b_1)$.

\begin{theorem}
For $n \ge 3$, the affine transformation group $\mathrm{AGL}(n) = \mathrm{AGL}(n, \mathbb{F}_q)$ for $q\geq 2$ can be generated by two elements $A=(I + E_{n1}+ (\alpha-1)E_{22},e_1 )$ and $B=( E_{12} + E_{23} + \ldots + E_{n1}, 0)$, where $e_i \in \mathbb{F}_q^n$ denotes the vector with entry 1 in row $i$, and $0$ elsewhere.
\end{theorem}
\begin{proof}
\textbf{Case 1. $q=2$}. Let $C=A^2=(I,e_n)$. Observe that $B^{-1} = (E_{1n} + E_{21} + E_{32} + \ldots + E_{n(n-1)}, 0)$, and $B^{-1} \circ (I, b) \circ B = (I, b')$, where $b' = (b_n, b_1, b_2, \ldots, b_{n-1})^T$. Thus
\[
(B^{-1})^j \circ (I, e_n) \circ B^j = (I, e_{j})
\]
for $j \in \{1, \ldots, n\}$. So we can generate $(I, e_j)$ for any $j$.
    
For any distinct $i_1, \ldots, i_k \in \{1, 2, \ldots, n\}$, we have
$
(I, e_{i_1}) \circ (I, e_{i_2}) \circ \ldots \circ (I,e_{i_k}) = (I,e_{i_1}+e_{i_2}+\ldots +e_{i_k}).
$
In this way, we can generate affine transformation $(I, b)$ for any $b \in \mathbb{F}_2^n$.

Finally, let us prove we can generate $(M, b)$ for any $(M, b) \in \agl(n, \mathbb{F}_2)$. Let $A' = A \circ (I, e_1) = (I + E_{n1},0)$. By Theorem \ref{thm:two_gen}, $A'$ and $B$ can generate $(M, 0)$ for any $M \in \mathrm{GL}(n,\mathbb{F}_2)$. In the last paragraph, we have shown that $(I, b)$ can be generated, so is $(I, b) \circ (M, 0) = (M, b)$.

\textbf{Case 2. $q\geq 3$.} We proceed in the following 5 steps. 

\textbf{Step 1:}
Let $C=(I+E_{n1},e_1)$ and $D=(I+(\alpha-1)E_{22},0)=(\diag(1,\alpha,1,...,1),0)$. Observe that 
    \[
    C \circ D = D \circ C = A,
    \]
    which implies that $C^q \circ D^q = A^q$. On the other hand, using induction we can prove 
    \[
    C^i = (I + iE_{n1}, ie_1 + (1 + 2 + \ldots + (i-1)) e_n)
    \]
    for any $i$. When $q=2^r$ for a positive integer $r\ge 2$, we know $C^4 = (I, 0)$, which implies $C^{q}={(C^4)}^{2^{r-2}}= C^4 = (I, 0)$; when the characteristic of $\mathbb{F}_q$, denoted by $p$, is at least $3$, we have $C^p = (I, 0)$, which implies that $C^q = (I, 0)$. So, $C^q = (I, 0)$ always holds for any $q \ge 3$. Combining with the fact that $C^q \circ D^q = A^q$, we have $D^q = A^q$. Note that $D$ is a diagonal matrix, so $D^q = D$, and thus $A^q = D$. So, we can generate $D$. Since $C = A \circ D^{-1}$,  we can generate $C$ as well.


\textbf{Step 2:}
The powers of $D = (\diag(1,\alpha,1,...,1),0)$ generate all elements of the form
\[
(\diag(1,\lambda,1,...,1),0)
\]
for all $\lambda \in \mathbb{F}_q^\times$. Note that
\[
B^{-1} \circ (\diag(a_1, a_2, \ldots, a_n), 0) \circ B = (\diag(a_n, a_1, \ldots, a_{n-1}), 0).
\]
So we can use $B$ to conjugate $D^i$ to move $\lambda$ to \emph{any} diagonal position. As such, we can generate all elements of the form
\[
(I + (\lambda - 1) E_{jj}, 0)
\]
for any $j \in \{1, 2, \ldots, n\}$, and for any $\lambda \in \mathbb{F}_q^\times$.

\textbf{Step 3:} Now let us prove we can generate $(I, \lambda e_n)$ for any $\lambda \in \mathbb{F}_q^\times$. Let $1 \not= \gamma\in \mathbb{F}_q^\times$, let
\begin{eqnarray*}
F & = & (I + (\gamma-1)E_{11}, 0) \circ C \\
& = & (I + (\gamma-1)E_{11}+E_{n1}, \gamma e_1).
\end{eqnarray*}
Using induction on $i$, one can prove
\[
F^i = (I + (\gamma^i-1)E_{11} + (\sum_{j=0}^{i-1} \gamma^j) E_{n1}, (\sum_{j=1}^i \gamma^j) e_1 + (\sum_{j=1}^{i-1} (i-j)\gamma^j)e_n).
\]
Consider $F^{q-1}$, which is
\[
(I + (\gamma^{q-1}-1)E_{11} + (\sum_{j=0}^{q-2} \gamma^j) E_{n1}, (\sum_{j=1}^{q-1} \gamma^j) e_1 + (\sum_{j=1}^{q-2} (q-j-1)\gamma^j)e_n).
\]
Note that $\gamma^{q-1} = 1$, $\sum_{j=0}^{q-2} \gamma^j = \frac{\gamma^{q-1} - 1}{\gamma - 1} = 0$, as long as $1 \not= \gamma \in \mathbb{F}_q^\times$. Similarly, $\sum_{j=1}^{q-1} \gamma^j = \gamma \sum_{j=0}^{q-2} \gamma^j = 0$, and 
$
\sum_{j=1}^{q-2} (q-j-1) \gamma^j = \frac{\gamma}{\gamma-1}.
$
Thus,
\[
F^{q-1} = (I, \frac{\gamma}{\gamma-1} e_n).
\]
When $\gamma \not= 1$, $\frac{\gamma}{\gamma-1}$ can take any nonzero value in $\mathbb{F}_q^\times$ except 1.

In Step 2, we have shown that elements $(\diag(1, \ldots, 1, (\gamma-1)\gamma^{q-2}), 0)$ and $(\diag(1, \ldots, 1, \frac{\gamma}{\gamma-1}), 0)$ can be generated. So we can generate
\begin{eqnarray*}
& & (\diag(1, \ldots, 1, (\gamma-1)\gamma^{q-2}), 0) \circ F^{q-1} \circ (\diag(1, \ldots, 1, \frac{\gamma}{\gamma-1}), 0)\\
& = & (\diag(1, \ldots, 1, (\gamma-1)\gamma^{q-2}), e_n) \circ (\diag(1, \ldots, 1, \frac{\gamma}{\gamma-1}), 0)\\
& = & (I, e_n).
\end{eqnarray*}
Therefore, we have shown that all elements of the form $(I, \lambda e_n)$,
 where $\lambda \in \mathbb{F}_q^\times$, can be generated. Using $B$ to conjugate $(I, \lambda e_n)$, we can move $\lambda$ to any diagonal position.

\textbf{Step 4:} For any distinct $i_1, \ldots, i_k \in \{1, 2, \ldots, n\}$ and for any nonzero $\lambda_1, \ldots, \lambda_k \in \mathbb{F}_q^\times$, we have
\[
(I, \lambda_1 e_{i_1}) \circ \ldots \circ (I, \lambda_k e_{i_k}) = (I,\lambda_1 e_{i_1}+ \ldots + \lambda_k e_{i_k}).
\]
In this way, we can generate elements $(I, b)$ for any $b \in \mathbb{F}_q^n$.

\textbf{Step 5:} Finally, let us prove we can generate $(M, b)$ for any $(M, b) \in \agl(n)$. Let 
\[
A' = (I, -e_1) \circ A = (I + E_{n1} + (\alpha-1)E_{22}, 0).
\]
By Theorem \ref{thm:two_gen}, $A'$ and $B$ can be generate $(M, 0)$ for any $M \in \mathrm{GL}(n,\mathbb{F}_q)$. In the last step, we have shown that $(I, b)$ can be generated, so is $(I, b) \circ (M, 0) = (M, b)$.
\end{proof}

Next, we prove the affine transformation group $\agl(2, \mathbb{F}_q)$ can be generated by two elements. Let $\alpha$ denote a generator of the multiplicative group of $\mathbb{F}_q$, and let $\beta$ denote a generator of the multiplicative group over a quadratic extension of $\mathbb{F}_q$. 

\begin{theorem}
For $q \ge 3$, the affine transformation group $\mathrm{AGL}(2, \mathbb{F}_q)$ can be generated by two elements $A=\left(\begin{pmatrix} 0 & -\alpha \\ 1 & \beta+\beta^{q} \end{pmatrix},0 \right)$ and $B=\left(\begin{pmatrix} \alpha & 0 \\ 0 & 1 \end{pmatrix},e_2 \right)$.
\end{theorem}
\begin{proof}
\textbf{Step 1.} We shall prove the two elements can generate affine transformation $(M, 0)$ for any $M \in \mathrm{GL}(2,\mathbb{F}_q)$. Using induction on $i$, one can easily prove
\begin{equation}
\label{equ:Bi_induction}
B^i = (I + (\alpha^i - 1) E_{11}, ie_2).    
\end{equation}
Let $C = B^q = (I + (\alpha-1)E_{11}, 0)$. By Theorem \ref{thm1:two_gen}, we know element $(M, 0)$, for any $M \in \mathrm{GL}(2, \mathbb{F}_q)$, can be generated.

\textbf{Step 2.} We will prove $(I, b)$, for any $b \in \mathbb{F}_q^{\times 2}$, can be generated. Let $D = B^{(q-1)(q-1)} = (I, e_2)$ by \eqref{equ:Bi_induction}.

Since all elements of the form $(M, 0)$, $M \in \mathrm{GL}(2, \mathbb{F}_q)$, can be generated, $F$ and $F^{-1}$ can be generated, where 
$$
F = \left(\begin{pmatrix} 1 & \alpha^i \\ 0 & \alpha^j \end{pmatrix},0 \right).
$$
Note that
$$
F \circ D \circ F^{-1} = \left(I, \begin{pmatrix} \alpha^i \\ \alpha^j \end{pmatrix}\right).
$$
So all elements of the form $(I, b)$, for any $b \in \mathbb{F}_q^{\times 2}$, can be generated.

\textbf{Step 3.} We will prove elements $(I, \lambda e_1)$ and $(I, \lambda e_2)$, for any $\lambda \in \mathbb{F}^\times_q$, can be generated. Since $\alpha$ is a generator, we have $\lambda = \alpha^i$ for some integer $i$. One can easily verify that 
\begin{eqnarray*}
& & \left(\begin{pmatrix} 1 &0 \\ -1 & 1 \end{pmatrix},0\right) \circ \left(\begin{pmatrix} 1 & 0 \\ 0 & 1 \end{pmatrix},  \begin{pmatrix} \alpha^i \\ \alpha^i  \end{pmatrix}\right) \circ \left(\begin{pmatrix} 1 &0 \\ 1 & 1 \end{pmatrix},0\right) \\
& = & \left(\begin{pmatrix} 1 & 0 \\ 0 & 1 \end{pmatrix},  \begin{pmatrix} \alpha^i \\ 0 \end{pmatrix}\right)
\end{eqnarray*}
and
\begin{eqnarray*}
& & 
\left(\begin{pmatrix} 1 &-1 \\ 0 & 1 \end{pmatrix},0\right) \circ \left(\begin{pmatrix} 1 & 0 \\ 0 & 1 \end{pmatrix},  \begin{pmatrix} \alpha^i \\ \alpha^i  \end{pmatrix}\right) \circ \left(\begin{pmatrix} 1 &1 \\ 0 & 1 \end{pmatrix},0 \right) \\
& = & \left(\begin{pmatrix} 1 & 0 \\ 0 & 1 \end{pmatrix},  \begin{pmatrix} 0\\ \alpha^i \end{pmatrix} \right).
\end{eqnarray*}
\textbf{Step 4.} From Step 2 and Step 3, we know element $(I, b)$, for any $b \in \mathbb{F}_q^2$, can be generated. In Step 1, we have shown that all elements of the form $(M, 0)$ can be generated, so is $(I, b) \circ (M, 0) = (M, b)$.
\end{proof}

As for $q=2$, $\mathrm{AGL}(2, \mathbb{F}_2)$ can be generated, for instance, by two elements $ \left(\begin{pmatrix} 0 & 1 \\ 1 & 1 \end{pmatrix},\begin{pmatrix} 0  \\ 0  \end{pmatrix} \right) $ and $\left(\begin{pmatrix} 1 & 1 \\ 0 & 1 \end{pmatrix},\begin{pmatrix} 1  \\ 0  \end{pmatrix} \right)$.

\begin{theorem}
	The affine transformation group $\mathrm{AGL}(n) = \mathrm{AGL}(n, \mathbb{F}_q)$ for $n\ge 2$ cannot be generated by one element.
\end{theorem}
\begin{proof}
A group generated by one element is a cyclic group and every subgroup of a cyclic group is also cyclic. For every finite group $G$ of order $n$, if $G$ is cyclic then $G$ has at most one subgroup of order $d$ for every divisor $d$ of $n$. Since $\vert \mathrm{AGL}(n,\mathbb{F}_q)\vert = q^n  \cdot \prod_{i=0}^{n-1} (q^n-q^i)=q^{2n-1}\cdot (q-1) \prod_{i=0}^{n-2} (q^n-q^i)$, then $q-1 $ is a divisor of $\vert \mathrm{AGL}(n,\mathbb{F}_q)\vert$ and so is $q$. If $ \mathrm{AGL}(n,\mathbb{F}_q)$ is a cyclic group, then $\mathrm{AGL}(n,\mathbb{F}_q)$ has at most one subgroup of order $q$ or $q-1$. In the following, we will prove there is more than one subgroup of order $q$ or $q-1$ in $\mathrm{AGL}(n,\mathbb{F}_q)$. 

Let $\alpha$ denote a generator of the multiplicative group $\mathbb{F}_q$. For $q=2$, let $A=(I+E_{12},0)$ and $B=(I+E_{21},0)$. The group $\langle A \rangle=\{(I+E_{12},0),(I,0)\}$ is a cyclic subgroup of order 2; the group $\langle B \rangle=\{(I+E_{21},0),(I,0)\}$ is a cyclic subgroup of order 2. For $q\ge 3$, let $C=(I+ (\alpha-1)E_{11},0)$ and $D=(I+ (\alpha-1)E_{22},0)$. The group $\langle C \rangle=\{(I+(\alpha-1)E_{11},0),(I+(\alpha^2-1)E_{11},0),\ldots,(I+(\alpha^{q-1}-1)E_{11},0)\}$ is a cyclic subgroup of order $q-1$; the group $\langle D\rangle=\{(I+(\alpha-1)E_{22},0),(I+(\alpha^2-1)E_{22},0),\ldots,(I+(\alpha^{q-1}-1)E_{22},0)\}$ is also a cyclic subgroup of order $q-1$. Therefore, $\mathrm{AGL}(n,\mathbb{F}_q)$ for $n\ge 2$ cannot be a cyclic group, i.e., the affine transformation group $\mathrm{AGL}(n,\mathbb{F}_q)$ cannot be generated by one element.
\end{proof}

\appendix

\section{Properties of $fn_i$ in the classification of $\rmc(6,6)/\rmc(3,6)$ }
\label{appendix_properties_fni}

We list some relevant algebraic properties of $fn_i$ in the classification of $\rmc(6,6)/\rmc(3,6)$, which are used in our proof. We obtain these numbers with the assistance of a computer, and all the C++ codes are available online.
\begin{table}[h!]
	\begin{center}
		\begin{tabular}{ c|c|c|c|c|c|c } 
			\toprule 
			$f$ & $fn_0$ & $fn_1$ & $fn_2$ & $fn_3$ & $fn_4$ & $fn_5$ \\ 
			\hline
			$\deg(f)$ & 0 & 4 & 4 & 4 & 5 & 5\\
			\hline
			${\nl}_2(f)$ & 0 & 4 & 6 & 10 & 2 & 4  \\ 
			\hline
			${\nl}_3(f)$ & 0 & 4 & 6 & 8 & 2 & 4  \\ 
			\hline
			$\ml_2(f)$ & 0 & 16 & 16 & 14 & 16 & 14  \\ 
			\bottomrule 
		\end{tabular}
		\begin{tabular}{ c|c|c|c|c|c } 
			\toprule 
			$f$ & $fn_6$ & $fn_7$ & $fn_8$ & $fn_9$ & $fn_{10}$ \\ 
			\hline
			$\deg(f)$  & 5 & 6 & 6 & 6 & 6 \\ 
			\hline
			${\nl}_2(f)$ & 8 & 1 & 3 & 7 & 9 \\ 
			\hline
			${\nl}_3(f)$ & 6 & 1 & 3 & 5 & 7 \\ 
			\hline
			$\ml_2(f)$  & 14 & 17 & 15 & 15 & 15 \\ 
			\bottomrule
		\end{tabular}
		\caption{Degree and nonlinearities of $fn_i$.}
		\label{table:fn10_nl3}
	\end{center}
\end{table}

\begin{table}[h!]
\begin{center}
\begin{tabular}{ c|c|c|c|c|c|c } 
\hline
$k$ & $6$ & $8$ & $10$ & $12$ & $14$ & $16$ \\ 
\hline
$|\mathcal{F}_{fn_2}^{(3)}(k)|$ & 64 & 1920 & 64320 & 579072 & 397440 & 5760 \\ 
\hline
$|\mathcal{F}_{fn_3}^{(3)}(k)|$ & 0 & 2304 & 71680 & 628992 & 345600 & 0 \\
\hline
$|\mathcal{F}_{fn_6}^{(3)}(k)|$ & 32 & 2112 & 65312 & 638208 & 342912 & 0 \\ 
\hline
\end{tabular}
\caption{Distribution of $|\mathcal{F}_{fn_i}^{(3)}(r)|$, $i \in \{2, 3, 6\}$.}
\label{table:Ff2_distribution}
\end{center}
\end{table}

\begin{table}[h!]
\begin{center}
\begin{tabular}{ c|c|c|c|c|c|c } 
\hline
$k$ & $5$ & $7$ & $9$ & $11$ & $13$ & $15$ \\ 
\hline
$|\mathcal{F}_{fn_9}^{(3)}(k)|$ & 6 & 298 & 12540 & 245556 & 784416 & 5760 \\ 
\hline
$|\mathcal{F}_{fn_{10}}^{(3)}(k)|$ & 0 & 288 & 13216 & 254016 & 746496 & 34560 \\
\hline
\end{tabular}
\caption{Distribution of $|\mathcal{F}_{fn_i}^{(3)}(r)|$, $i \in \{9, 10\}$.}
\label{table:Ff9_distribution}
\end{center}
\end{table}

Table \ref{table:orbit_length} summarizes the $\mathrm{AGL}(6)$-orbit length $fn_i + \rmc(3,6)$ in the quotient space $\rmc(6,6)/\rmc(3, 6)$.
\begin{table}[h!]
\begin{center}
\begin{tabular}{ c|c|c|c|c|c|c } 
\hline
 $f$ & $fn_0$ & $fn_1$ & $fn_2$ & $fn_3$ & $fn_4$ & $fn_5$  \\
\hline
orbit length & 1 & 651 & 18228 & 13888 & 2016 & 312480  \\ 
\hline
\end{tabular}
\begin{tabular}{ c|c|c|c|c|c } 
\hline
 $f$ & $fn_6$ & $fn_7$ & $fn_8$ & $fn_9$ & $fn_{10}$  \\
\hline
orbit length & 1749888 & 64 & 41664 & 1166592 & 888832 \\ 
\hline
\end{tabular}
\caption{Orbit length of $fn_i + \rmc(3, 6)$.}
\label{table:orbit_length}
\end{center}
\end{table}

\section*{Acknowledgements}
	We thank the editor and anonymous reviewers for their valuable comments, for example, for their suggestions to present the results in more general forms.


\printbibliography

\end{document}